\documentclass[11pt, a4paper, twoside]{amsart}

\theoremstyle{plain}
\newtheorem{theorem}{Theorem}[section]
\newtheorem{proposition}[theorem]{Proposition}

\theoremstyle{definition}

\newcommand {\Set}[1] {\mathbb{#1}}
\newcommand{\setR}[0]{\Set{R}}

\newcommand{\slaz}[0]{\setminus \{0\}}

\newcommand{\ttt}[0]{\widetilde}

\newcommand{\thCite}[1]{\emph{#1}}

\newcommand{\pd}[2]{\frac{\partial #1}{\partial #2}}
\newcommand{\pdd}[3]{\frac{\partial^2 #1}{\partial #2\,\partial #3}}

\newcommand{\od}[2]{\frac{d #1}{d #2}}

\title{On the tensorial properties of the generalized Jacobi equation}

\author[Dahl]{Matias F. Dahl}
\address{
Matias Dahl,
Aalto University,
Institute of Mathematics,
P.O. Box 11100,
FI-00076 Aalto,
Finland
}
\email{matias.dahl@aalto.fi}
\author[Gallego Torrom\'e]{Ricardo Gallego Torrom\'e}
\address{
Ricardo Gallego Torrom\'e,
Instituto de matem\'atica e estat\'istica - USP, S\~ao Paulo, Brazil}
\email{rgallegot@gmx.de}

\begin{document}

\maketitle
\begin{abstract}
  The generalized Jacobi equation is a differential equation in local coordinates that describes the
  behavior of infinitesimally close geodesics with an arbitrary relative velocity.  In this note we
  study some transformation properties for solutions to this equation.  We prove two results.  First,
  under any affine coordinate changes we show that the tensor transformation rule maps solutions to
  solutions.  As a consequence, the generalized Jacobi equation is a tensor equation when restricted
  to suitable Fermi coordinate systems along a geodesic.  Second, in dimensions $n\ge 3$, we
  explicitly show that the transformation rule does not in general preserve solutions to the
  generalized Jacobi equation.
\end{abstract}
\bigskip

\section{Introduction}
\label{sec:Introduction}
Suppose $M$ is a manifold with an affine and torsion free connection $\nabla$.  In this setting, the
{\it Jacobi equation} is a fundamental equation that describes the qualitative behavior of
infinitesimally close geodesics on $M$. See for example, \cite{Levi-Civita} and \cite{Hicks}.  One
way to derive the Jacobi equation is to consider the deviation $\xi^\mu(s)= x^\mu(s)-X^\mu(s)$
between two neighboring geodesic $X$ and $x$. Then the Jacobi equation follows by subtracting the geodesic
equations for $X$ and $x$ and assuming that $\xi^\mu(s)$ and $\dot \xi^\mu(s)=\od{\xi^\mu}{s}(s)$
are both infinitesimal quantities. This is a standard argument. See for example \cite{Perlick07} and
references therein.  The \emph{generalized Jacobi equation} is derived in the same way, but under
the weaker hypothesis that only $\xi^\mu(s)$ is an infinitesimal quantity. That is, the geodesics
are assumed to be infinitesimally close, but their relative velocity $\dot \xi^\mu$ does not need to
be small. Under this weaker hypothesis, equations for the displacement $\xi^\mu$ was derived
first in the Lorentzian case by Hodgkinson \cite{Hodgkinson} and independently by Ciufolini
\cite{Ciufolini}.  This generalization of the Jacobi equation have been investigated by several
authors, specially with applications on astrophysics and cosmology. See for instance
\cite{Mashhoon1, ChiconeMashhoon:2002, Perlick07}.

We will work in the setting of an affine and torsion-free connection. In this setting, the
generalized Jacobi equation was derived by Perlick \cite{Perlick07}.  In more detail, a collection
of $n$ functions $\xi^{\mu}\colon I\to \setR$ along a geodesic $X\colon I\to M$ is a solution to the
\emph{generalized Jacobi field} in coordinates $\{x^\mu\}_{\mu=0}^{n-1}$ if
\begin{eqnarray}
\label{eq:genJacobiEq}
\ddot \xi^{\mu}
+ \Gamma^{\mu}_{\nu\rho} \left( 2 \dot \xi^{\rho} \dot X^{\nu} + \dot \xi^{\rho} \dot \xi^{\nu}\right)
+ \pd{\Gamma^{\mu}_{\rho \nu}}{x^\tau} \xi^\tau \left( \dot X^{\rho}
 + \dot \xi^{\rho}\right) \left( \dot X^{\nu} + \dot \xi^{\nu}\right) &=& 0,
\end{eqnarray}
where dots indicate ordinary derivatives with respect to the parameter $s$ of the central geodesic
$X^\mu(s)$ and $\Gamma^\mu_{\nu\rho}$ are the connection coefficients of $\nabla$.

Let us first observe that in contrast to the usual Jacobi equation, the generalized Jacobi equation
is a non-linear equation in the unknown functions $\xi^\mu$.  This makes the analysis of the
solution space more difficult.  For example, the solution space need not be a vector space and in
general, there are no (known) results that associate a solution $\xi^\mu$ to a geodesic variation as
for the Jacobi equation.
%Numerical experiments \cite{ChiconeMashhoon:2002} show that solutions to
%equation \eqref{eq:genJacobiEq} exhibit a chaotic behavior.

The generalized Jacobi equation \eqref{eq:genJacobiEq} is an equation in local coordinates.  If $\nabla$ is \emph{flat},
there exits coordinates where $\Gamma^{\mu}_{\nu\rho}=0$ \cite[Proposition 1.1]{Shima:2007}.  Then
the generalized Jacobi equation and the usual Jacobi equation both simplify into $\ddot
\xi^{\mu}=0$. More generally, Perlick \cite{Perlick07} has proven that if $X$ is a lightlike
geodesic in a special class of Lorentz metrics, or \emph{planewave metrics}, then there are coordinates
around $X$ where functions $\xi^\mu$ satisfy equation \eqref{eq:genJacobiEq} if and only if
functions $\xi^\mu$ satisfy the usual Jacobi equation.
When this is the case, it implies that by a suitable choice of coordinates, the nonlinear
generalized Jacobi equation can be replaced by the linear Jacobi equation.  Since the Jacobi
equation transforms tensorially, this motivates a further understanding of the transformation
properties of the generalized Jacobi equation. Namely:
\begin{enumerate}
\item [] Suppose $X\colon I\to M$ is a geodesic in overlapping coordinates
  $x^\mu$ and $\widetilde x^\mu$. If functions $\xi^\mu\colon I\to \setR$ solve equation
  \eqref{eq:genJacobiEq} in coordinates $x^\mu$, does there exist a
  transformation rule $\xi^\mu\mapsto \widetilde \xi^\mu$ such that functions $\widetilde \xi^\mu$
  solve equation \eqref{eq:genJacobiEq} in coordinates $\widetilde x^{\mu}$?
\end{enumerate}
In this note we will study this question for the tensorial transformation rule $\widetilde\xi^\mu =
\pd{\widetilde x^\mu}{x^\nu} \xi^\nu$.  We will prove two results.  First, in \thCite{proposition}
\ref{thm:pos} we show that the tensorial transformation rule for $\xi^i$ preserves solutions to
equation \eqref{eq:genJacobiEq} for affine coordinate changes.  As a consequence, equation
\eqref{eq:genJacobiEq} is a tensorial equation when restricted to suitable Fermi coordinate systems
along a geodesic (see Proposition \ref{prop:FermiTrans}).  This motivates the use of Fermi
coordinates for the study of the generalized Jacobi equation as in \cite{ChiconeMashhoon:2002,
  Perlick07}.
Second, in \thCite{proposition} \ref{thm:neg} we explicitly show that in
dimensions $n\ge 3$, the tensorial transformation rule for $\xi^{\mu}$ does not in general preserve solutions
to equation \eqref{eq:genJacobiEq}. Thus, if there exists a
transformation rule $\xi^\mu\to \widetilde \xi^\mu$ for solutions to equation
\eqref{eq:genJacobiEq}, it is not the tensorial transformation rule.

The organization of this note is as follows.  In Section \ref{sec:Preliminaries} we review the
necessary theory for affine connections and Fermi coordinates. In Section
\ref{sec:EqsGeodesicDeviation} we summarize the derivation of the Jacobi equation and the
generalized Jacobi equation \cite{Levi-Civita}. Lastly, in Section \ref{sec:GJE-transformation} we
prove \thCite{proposition} \ref{thm:pos} and \thCite{proposition} \ref{thm:neg} described above.

\section{Preliminaries}
\label{sec:Preliminaries}
Let $M$ be a smooth manifold of dimension $n \ge 2$. By $TM$ we denote the tangent bundle with
projection $\pi\colon TM\to M$.  The tangent space at a point $p\in M$ is defined as
$T_pM=\pi^{-1}(p)$.  Throughout the paper we denote by $I$ an interval in $\setR$. We will also use
the Einstein summing convention.

We assume that $M$ is endowed with an \emph{affine connection} $\nabla$. Thus, in each
coordinate chart $(U, x^\mu)$, $\nabla$ is represented by connection coefficients
 $\Gamma^\mu_{\nu \sigma}$ and if
% $\{\Gamma^\mu_{\nu \sigma},\,\mu,\nu,\sigma=0,...,n-1\}$ and if
 $\Gamma^\mu_{\nu \sigma}$ and  $\widetilde \Gamma^\mu_{\nu \sigma}$ represent $\nabla$ on
overlapping coordinates $x^\mu$ and $\widetilde x^\mu$, we have transformation rules
\begin{eqnarray}
\label{eq:gammaSymbTrans}
\pd{\widetilde x^\lambda}{x^\alpha}    \Gamma^\alpha_{\nu\sigma}
&=&
\widetilde \Gamma^\lambda_{\alpha\beta}   \pd{\widetilde x^\alpha}{x^\nu}  \pd{\widetilde x^\beta}{x^\sigma}
+\pdd{\widetilde x^\lambda}{x^\nu}{x^\sigma}.
\end{eqnarray}
A connection is \emph{torsion-free} if  $\Gamma^\mu_{\nu \sigma}=\Gamma^\mu_{\sigma\nu}$.
If $X$ is a curve $X\colon I\to M$, then $X$ is an \emph{(affinely parameterized) geodesic}  if
\emph{(i)} $X$ is a \emph{regular curve}, that is, the tangent $\dot X$ is never zero, and
\emph{(ii)} in each coordinate chart  $(U,x^{\mu})$ that overlaps $X$ we have
\begin{eqnarray}
\label{eq:geodEq}
\ddot X^\mu + \Gamma^\mu_{\nu\sigma}(X) \dot X^\nu \dot X^\sigma &=& 0,
\end{eqnarray}
where $(X^\mu(s))$ are components for $X$ in coordinates $x^\mu$.
% local coordinate system for $M$ such that on $U$
%Suppose $(U,x^\lambda)$ and $(\widetilde U, \widetilde x^\lambda)$ are overlapping coordinates and
%$U\cap \widetilde U\cap X(I)\,\neq \emptyset$.  Then equation \eqref{eq:gammaSymbTrans} implies that
%equation \eqref{eq:geodEq} holds in coordinates $U$ if and only if equation \eqref{eq:geodEq} holds
%in coordinates $\widetilde U$.

\subsection{Fermi coordinates}
\label{sec:FermiCoords}
\newcommand {\iInt}[0] {\mathcal{I}} In this section we collect some results on  Fermi
coordinates. These are local coordinates for a tubular neighborhood around a geodesic
\cite{ManasseMisner:1963}.  Suppose $\iInt\subset \setR$ is an open interval and
$X\colon \iInt{}\to M$ is a geodesic of
an affine connection on $M$,
and suppose $\Pi$ is an $(n-1)$-dimensional vector space in $T_{X(s_0)}M$
for some $s_0\in \iInt{}$ such that $\Pi$ is complementary to $\dot X(s_0)$, that is,
\begin{eqnarray}
\label{eq:TXS0Span}
   T_{X(s_0)} M &=& \operatorname{span}\{ \dot X(s_0)\} \,\,\oplus \,\, \Pi.
\end{eqnarray}
Let $\{e_1, \ldots, e_{n-1}\}$ be a basis for $\Pi$. By parallel transport we can extend each vector
$e_i\in T_{X(s_0)}M$ into a vector field $e_i(s)$ along $X$. Since parallel transport is a linear
isomorphism, it follows that
\begin{eqnarray}
\label{eq:TXS0SpanTr}
T_{X(s)} M &=& \operatorname{span}\{ \dot X(s), e_1(s), \ldots, e_{n-1}(s)\}
\end{eqnarray}
for all $s\in \iInt{}$. For $s\in \iInt{}$, let $f$ be the map
\begin{eqnarray}
\label{eq:fermiDef}
 f(s, z^1, \ldots, z^{n-1}) &=& \exp\left\{\sum_{i=1}^{n-1}
%w^i\left.\pd{}{z^i}\right\vert_{X(\sigma)}\right\}
z^{i} e_{i}(s) \right\}
\end{eqnarray}
defined for $z^1, \ldots, z^{n-1}\in \setR$ for which the right hand side
is defined.
\begin{proposition}
\label{prop:Fermi}
Suppose $\iInt{}$ is an open interval, %and %$I^\ast\subset$
$X\colon \iInt{} \to M$ is a geodesic for an affine connection $\nabla$, and
$\{e_{i}\}_{{i}=1}^{n-1}$ and $f$ are as in equation \eqref{eq:fermiDef}.  Moreover, suppose $X$ has no
self-intersections, and $I$ is a proper open subset $I\subset \iInt{}$ such that $\overline{I}$ is
compact.
Then there exists an open neighborhood of the origin $B\subset \setR^{n-1}$ such that $f$ restricts
to a diffeomorphism $f\colon I\times B \to M$ onto its range.
\end{proposition}
\begin{proof}
  Since $\overline{I}$ is compact, we can find an open ball $B\subset \setR^{n-1}$ containing $0$
  such that $f\colon I\times B\to M$ is smooth.  For $s\in I$ we have
\begin{eqnarray*}
  \pd{f}{s}(s,0,\ldots, 0)&=&\dot{X}(s),\\
  \pd{f}{z^{i}}(s,0,\ldots, 0)&=& e_{i}(s), \quad i\in \{1, \ldots, n-1\}.
\end{eqnarray*}
Since $\{\dot{X}(s), e_1(s), \ldots, e_{n-1}(s)\}$ are linearly independent for all $s\in I$,
%The condition of torsion-free connection assures that the vectors $\{\dot{X}(s,0...,0)$, $
%e_1(s,0,...,0)$, $\ldots, e_{n-1}(s,0,...,0)\}$ are linearly independent in some interval of time
%containing $s\in I_0$. Then
the inverse function theorem implies (after possibly shrinking $B$) that
$f\colon I\times B\to M$ is a local diffeomorphism onto its range.  The result follows by
\cite[p.~345, Lemma 19]{spivakI}.
\end{proof}
When $(s,z^1, \ldots, z^{n-1})\in I\times B$ are as in Proposition \ref{prop:Fermi} we say that
$(s,z^1, \ldots, z^{n-1})$ are \emph{Fermi coordinates} along $X\colon I \to M$. These
coordinates are determined by the
%by the initial hyperplane $\Pi_{X(s_0)}\subset\, T_{X(s_0)}M$.  Therefore, a Fermi
%coordinate system is determined by an
geodesic $X\colon I\to M$, the initial point $X(s_0)$ and the set of vectors
$\{e_i(s_0)\}^{n-1}_{i=1}$ in equation \eqref{eq:TXS0SpanTr}.
One can prove that when the connection is torsion-free, then all Christoffel symbols
$\Gamma^\mu_{\nu\rho}$ vanish on the central geodesic in Fermi coordinates.
The next proposition shows that if two Fermi coordinates systems are determined by the same initial
complementary hyperplane $\Pi$, then the Fermi coordinates differ by an affine coordinate
transformation. Here, two overlapping coordinates $(U, x^{\mu})$ and $(\widetilde U, \widetilde
x^\mu)$ are related by an \emph{affine coordinate transformation} if
\begin{eqnarray}
\label{eq:affineTrans}
\widetilde x^{\mu}(x^0, \ldots, x^{n-1}) &=& \Lambda^{\mu}_{\nu} x^{\nu} + C^{\mu}
\end{eqnarray}
for some constants $(\Lambda^{\mu}_{\nu})_{\mu,\nu=0}^{n-1}$ and $(C^{\mu})_{\mu=0}^{n-1}$.

\begin{proposition}
\label{prop:FermiTrans}
Suppose $\nabla$ is an affine connection on $M$ and $X\colon I\to M$ and $\widetilde{X}\colon
\widetilde{I}\to \,M$ are two geodesics that differ by a reparameterization.
%such that $X(I)=\widetilde {X}(\widetilde I)$.
Moreover, suppose
$(s, z^1, \ldots, z^{n-1})$ and $(\widetilde{s}, \widetilde z^1, \ldots, \widetilde z^{n-1})$ are Fermi
coordinates along $X$ and $\widetilde X$ that correspond to the same initial complementary hyperplane $\Pi$.
%$\Pi_{X(s_0)}=\Pi_{\widetilde X(\widetilde s_0)}$.
Then Fermi coordinates
$(s, z^1, \ldots, z^{n-1})$ and $(\widetilde{s}, \widetilde z^1, \ldots, \widetilde z^{n-1})$ are related by an affine coordinate transformation on their common domain.
More precisely, there are constants $A,B$ and an invertible matrix
$(T^{\nu}_{\nu})_{\mu,\nu=1}^{n-1}$ such that
\begin{eqnarray}
\label{eq:FermiAffineRelation}
\widetilde s = As +B, \quad
\widetilde z^\mu = \sum_{\nu=1}^{n-1} T^\mu_\nu z^\nu\,\, \mbox{for}\,\,
\mu\in \{1, \ldots, n-1\}.
\end{eqnarray}
\end{proposition}
\begin{proof}
  Let $f$ and $\widetilde f$ be maps $f\colon I\times B\to M$ and $\widetilde f\colon \widetilde
  I\times \widetilde B\to M$ as in equation \eqref{eq:fermiDef} that define the two Fermi
  coordinates.
  Let $\psi\colon I\to \widetilde I$ be the reparameterization, so that $\widetilde X\circ \psi =
  X$. Writing out the geodesic equation for $\widetilde X$ and $\widetilde X\circ \psi$ shows that
  $\psi(s)=A s + B$ for some $A\in\setR\slaz$ and $B\in \setR$.
  By assumption, there are $s_0\in I$ and $\widetilde s_0\in \widetilde I$ such that
vectors $\{\pd{}{z^{\mu}}\vert_{X(s_0)}\}_{\mu=1}^{n-1}$ and $\{\pd{}{\widetilde
    z^{\mu}}\vert_{\widetilde X(\widetilde s_0)}\}_{\mu=1}^{n-1}$ span the same hyperplane.
Thus $X(s_0)=\widetilde X(\widetilde s_0)$ and
%since neither geodesic is self-intersecting it follows that
% $\widetilde s_0 = \psi(s_0)$. Then
there exist an invertible matrix $T=(T^\nu_\mu)_{\mu,\nu=1}^{n-1}\in \setR^{(n-1)\times (n-1)}$ such
that
\begin{eqnarray}
\left.    \pd{}{z^{\mu}}\right\vert_{X(s)} &=& \label{eq:jkasd}
\sum_{\nu=1}^{n-1} T_{\mu}^{\nu} \left.\pd{}{\widetilde z^{\nu}}\right\vert_{ X(s)}, \quad \mu\in \{1,\ldots, n-1\}
\end{eqnarray}
at $s=s_0$. Since both sides in this equation are parallel vectors along $X$, and since equality
holds for $s=s_0$, it follows that equation \eqref{eq:jkasd} holds for all $s\in I$.

Suppose $(s, z^1, \ldots, z^{n-1})$ and $(\tilde{s}, \widetilde z^1, \ldots,
\widetilde z^{n-1})$ are Fermi coordinates for the same point, so that
\begin{eqnarray}
\label{eq:FermiRelation1}
  f( s, z^1, \ldots, z^{n-1}) &=&   \widetilde f(\widetilde s, \widetilde z^1, \ldots, \widetilde z^{n-1}).
\end{eqnarray}
Then equations \eqref{eq:fermiDef} and \eqref{eq:jkasd} %and $\psi(s)=As+B$
imply that
\begin{eqnarray}
\label{eq:FermiRelation2}
 f( s,  z^1, \ldots,  z^{n-1}) &=&
 \widetilde f\left(As+B, \sum_{\mu=1}^{n-1} T^{1}_{\mu}z^\mu, \ldots,
\sum_{\mu=1}^{n-1} T^{n-1}_{\mu} z^\mu\right).
\end{eqnarray}
Since $f$ is a bijection onto its range, equations
\eqref{eq:FermiRelation1}--\eqref{eq:FermiRelation2} imply that
equation \eqref{eq:FermiAffineRelation} holds.
%
%$\widetilde s = As +B$ and
%$\widetilde z^\mu = \sum_{\nu=1}^{n-1} T^\mu_\nu z^\nu$ for all
\end{proof}

Let us consider the case when $(M,g)$ is a pseudo-Riemann manifold of index $1$ and $\nabla$ is the
Levi-Civita connection of $g$.  If $\dot{X}(s_0)$ is timelike or spacelike, then a suitable
hyperplane $\Pi$ is given by $\Pi = (\dot{X}(s_0))^\perp$ \cite[p.~49]{ONeill:1983}.  In this case,
the metric can further be made diagonal along the geodesic in Fermi coordinates
\cite{Levi-Civita}. On the other hand, if $\dot{X}(s_0)$ is lightlike, we have $\dot X(s_0)\in (
\dot X(s_0))^\perp$ and the choice $\Pi = (\dot{X}(s_0))^\perp$ is not possible.

\section{Equations of geodesic deviation}
\label{sec:EqsGeodesicDeviation}
In this section we describe three equations for the behavior of nearby geodesics: \emph{the exact
  deviation equation}, \emph{the Jacobi equation}, and \emph{the generalized Jacobi equation}.
Throughout this section we assume that $\nabla$ is an affine and torsion-free connection on $M$.

\subsection{The exact geodesic deviation equation}
\label{sec:ExactDeviation}
Suppose $x\colon I\to M$ and $ X\colon I\to M$ are two geodesics that are contained in one
coordinate chart $(U, x^\mu)$. If locally $x(s)=(x^\mu(s))$ and $X(s)=(X^\mu(s))$ for $s\in I$, let
$\xi\colon I\to \setR^n$ be the displacement between the geodesics defined as
\begin{eqnarray}
\label{eq:xiDef}
 \xi^\mu(s) &=& x^\mu(s) - X^\mu(s), \quad s\in I.
\end{eqnarray}
Since $x$ and $X$ are  solutions to the geodesic equation, it follows that
%\begin{eqnarray}
%\label{standardexactdeviationequation}
%\frac{d^2\xi^{\mu}}{ds^2}&+&\Gamma^{\mu}_{\nu\sigma}(X+\xi)\,\Big( \frac{dX^{\nu}}{ds}+ \frac{ \,d\xi^{\n%u}}{ds}\Big)
%\Big( \frac{dX^{\sigma}}{ds}+ \frac{\,d\xi^{\sigma}}{ds}\Big) \\
%& &\nonumber -\,\,\Gamma^{\mu }_{\nu\sigma}(X)\,\frac{d X^{\sigma}}{ds}\,\frac{d X^{\nu}}{ds}=\,0.
%\end{eqnarray}
\begin{eqnarray}
\label{standardexactdeviationequation}
\quad \quad \ddot\xi^{\mu}\,+\,\Gamma^{\mu}_{\nu\sigma}(X+\xi)\,
\Big( \dot X^{\nu}+ \dot \xi^{\nu} \Big)
\Big( \dot X^{\sigma}+ \dot\xi^{\sigma}\Big)
%& &\nonumber
 -\,\,\Gamma^{\mu }_{\nu\sigma}(X)\,
\dot X^{\sigma} \,\dot X^{\nu}=\,0.
\end{eqnarray}
We will refer to \eqref{standardexactdeviationequation} as the \emph{exact geodesic deviation
  equation} (see \cite{Perlick07}).
%We should emphasize that the definition of $\xi^\mu$ in equation \eqref{eq:xiDef} and
%the derivation of equation \eqref{standardexactdeviationequation} is done in local
%coordinates.
Equation \eqref{standardexactdeviationequation} is an exact equation in the sense that its
derivation does not involve any approximations. However, the geometric analysis of equations
\eqref{eq:xiDef}--\eqref{standardexactdeviationequation} becomes difficult since both equations are
defined in local coordinates and, moreover, both equations involve two points on the manifold.

%is difficult since functions $\xi^{\mu}(s)$ does not
%live over one point in $M$, but  for each value of the parameter $s$.  From the above considerations,
%it is clear that the exact deviation equation \eqref{standardexactdeviationequation} is quite
%inconvenient from a geometric point of view, since the functions $\xi^{\mu}$ are not {\it local},
%that is, do not depend on one point on $M$ but on two points, $x(s)$ and $X(s)$.

That components $\xi^\mu$ defined by equation \eqref{eq:xiDef} do not define a vector field along $X\colon
I\to M$ can be seen as follows. If $\setR^2$ is equipped with the Euclidean metric and Cartesian
coordinates $(x^1, x^2)$, then curves $X(s)=(s,0)$ and $x(s)=(s,1)$ for $s>0$ are geodesics and
functions $\xi^{\mu}\colon I\to \setR^2$ are given by $\xi(s) = (0,1)$.  However, in polar
coordinates $\widetilde x^1 = r, \widetilde x^2 = \theta$, the definition of $\xi$ yields
$\widetilde \xi(s) = (\sqrt{1+s^2}-s, \tan^{-1}(\frac{1}{s}))$. It follows that $\xi^1(s)=
\left(\frac{\partial x^1}{\partial \widetilde x^\mu}\right)\big|_{X(s)}\widetilde \xi^\mu$ is not
satisfied for all $s>0$, and functions $\xi^\mu$ in equation \eqref{eq:xiDef} do not in general
transformation as a tensor.

Even if functions $\xi^\mu(s)$ in equation \eqref{eq:xiDef} do not transform as a tensor in
general, it turns out that if we restrict to suitable Fermi coordinates along a geodesic,
%$X\colon I \to M$ as in
%{\it proposition} \ref{prop:FermiTrans},
then functions $\xi^0, \ldots, \xi^{n-1}$ transform as a vector.  To see this, suppose $X\colon
I\to M$ is a geodesic and $(s, z^1, \ldots, z^{n-1})$ are Fermi coordinates along $X$ as in {\it
  proposition} \ref{prop:Fermi}. If $x\colon I\to M$ is another geodesic that can be written as
$x(s)=(s,z^1(s), \ldots, z^{n-1}(s))$ in these Fermi coordinates, then functions $\xi^0, \ldots,
\xi^{n-1}$ in equation \eqref{eq:xiDef} are given by
\begin{eqnarray}
\label{eq:FermiXi}
\xi^\mu(s) &=&
\begin{cases}
0, & \mbox{for}\,\, \mu=0,\\
z^\mu(s), & \mbox{for}\,\, \mu\in \{1, \ldots, n-1\}.
\end{cases}
\end{eqnarray}
If $(\widetilde s, \widetilde z^1, \ldots, \widetilde z^{n-1})$ are other Fermi coordinates along
$X$ as in {\it proposition} \ref{prop:FermiTrans} and $\widetilde \xi^\mu$ are defined by equation
\eqref{eq:xiDef} in these coordinates, then equations \eqref{eq:fermiDef} and
\eqref{eq:FermiAffineRelation} show that $\xi^\mu$ and $\widetilde \xi^\mu$ transform as a vector
along $X$.

Since we have not made any approximations when deriving equation
\eqref{standardexactdeviationequation}, it turns out that we can transform solutions from one
coordinate system to another. However, this will lead to a highly non-standard transformation rule.
To see this, suppose $X^\mu(s)$ and $x^\mu(s)$ are geodesics as above in coordinates $x^\mu$ whence
$\xi^\mu(s)$ in equation \eqref{eq:xiDef} solve equation \eqref{standardexactdeviationequation}.
Thus, if $\widetilde x^\mu$ are overlapping coordinates and $\widetilde X^\mu(s)$ and $\widetilde
x^\mu(s)$ represent the same geodesics in these coordinates, then functions $\widetilde \xi^\mu(s)
= \widetilde x^\mu(s) - \widetilde X^\mu(s)$ solve equation \eqref{standardexactdeviationequation}
in $\widetilde x^\mu$-coordinates.  Then
\begin{eqnarray}
\label{eq:nonLinearTransForEGDE}
 \widetilde \xi^\mu(s)
%=  \widetilde x^\mu(s)  - \widetilde X^\mu(s)
& =& (\widetilde x\circ x^{-1})^\mu \left(\xi^\lambda(s) + X^\lambda(s)\right) - \widetilde X^\mu(s),
\end{eqnarray}
%where $(\widetilde x\circ x^{-1})^\mu$ is the transformation from $x^\mu$-coordinates to $\widetilde
%x^\mu$-coordinates,
and the above equation gives a transformation rule $\xi^\mu\mapsto \widetilde
\xi^\mu$ for solutions to the exact deviation equation along a geodesic $X\colon I\to M$.

\subsection{The Jacobi equation}
\label{sec:JacobiEq}
The usual approach to analyze the qualitative behavior of nearby geodesics is by using the \emph{Jacobi
  equation}, which describes the displacement $\xi^\mu$ between two geodesics under the
assumption that $\xi^\mu$ and $\dot \xi^\mu$ are infinitesimal \cite{Levi-Civita}.
%{\it infinitesimally close} .

Let us show how the Jacobi equation follows from equation \eqref{standardexactdeviationequation}.
Using Taylor's formula, we can expand each connection coefficients as
\begin{eqnarray}
\label{eq:gammaLin}
   \Gamma^{\mu}_{\nu\sigma}(X + \xi) &=&    \Gamma^{\mu}_{\nu\sigma}(X) + \pd{\Gamma^{\mu}_{\nu\sigma}}{x^{\tau}}(X) \xi^{\tau}\,+ \mbox{higher order terms.}
%\mathcal{O}(\xi^2).
\end{eqnarray}
Inserting equation \eqref{eq:gammaLin} into equation \eqref{standardexactdeviationequation} and
assuming that  $\xi^\mu$, $\dot \xi^\mu$ are infinitesimal
yields
\begin{eqnarray}
\ddot \xi^{\mu}\,+\pd{\Gamma^{\mu}_{\nu\sigma}}{x^\tau} (X)\,\xi^{\tau}
\dot X^{\nu}\,\dot X^{\sigma}
+2\,\Gamma^{\mu }_{\nu\sigma}(X)\,\dot X^{\sigma} \dot\xi^{\nu}&=&\,0.
\label{Jacobiequation}
\end{eqnarray}
That is, in the above we assume that functions $\xi^\mu$ are such that  all  higher order terms
$\xi^\mu \xi^\nu$,
$\dot \xi^\mu  \xi^\nu$,
$\dot \xi^\mu \dot \xi^\nu$,
$\xi^\mu \xi^\nu \xi^\sigma$,
$\dot \xi^\mu \xi^\nu \xi^\sigma$,
%$\dot \xi^\mu \dot \xi^\nu \xi^\sigma$,
$\ldots$
can be neglected.
Let us emphasize that due to this assumption on $\xi^\mu$,
we can no longer treat $\xi$ in equation \eqref{Jacobiequation} as the exact
displacement between two geodesics as in equation \eqref{eq:xiDef}.

Equation \eqref{Jacobiequation} is known as the \emph{Jacobi equation} for an affine and
torsion-free connection $\nabla$.  Using the covariant derivative $\frac{D}{Ds}$ along $X$, equation
\eqref{Jacobiequation} can equivalently be rewritten as %\cite{Levi-Civita}
\begin{eqnarray}
    \frac{D^2}{D s^2} \xi + R(\xi,\dot{X})(\dot{X}) &=& 0,
 \label{covariantformjacobiequation}
\end{eqnarray}
where $R(\xi,\dot{X})$ the curvature endomorphism of $\nabla$ and $\xi(s) = \xi^\mu(s)
\pd{}{x^\mu}\vert_{X(s)}$. See for example \cite{Perlick07}.  From equations
\eqref{Jacobiequation}--\eqref{covariantformjacobiequation} we see that the assumptions on $\xi^\mu$
and $\dot \xi^\mu$ have simplified the exact geodesic deviation equation in three significant ways:
First, unlike the exact geodesic deviation equation, equations \eqref{Jacobiequation} and
\eqref{covariantformjacobiequation} are linear equations in functions $\xi^\mu$.
Second, as we will see below, equations \eqref{Jacobiequation} and \eqref{covariantformjacobiequation}
will be covariant when functions $\xi^\mu$ transform as components of a tensor. This is a
simplification when compared with the nonlinear transformation rule \eqref{eq:nonLinearTransForEGDE} for solutions
to the exact geodesic deviation equation.
Lastly, equations \eqref{Jacobiequation} and \eqref{covariantformjacobiequation} are equations along
$X$ that only involve evaluations at $X(s)$ on $M$.

In the next sections we will study the coordinate invariance of the generalized Jacobi
equation.  As a model for this analysis and to fix notation, let us consider in some detail the
coordinate invariance for the Jacobi equation. Suppose $(U,x^\mu)$ are local coordinates for $M$ and
$\xi=\xi^\mu\left.\pd{}{x^\mu}\right\vert_{X(s)}$ is a curve $\xi\colon I\to TU$ along a geodesic
$X\colon I\to M$. For $\mu\in \{0,\ldots, n-1\}$  we define
\begin{eqnarray}
\label{eq:Jdef}
J^\mu_U[(\xi^\lambda)_{\lambda=0}^{n-1}] &=& \ddot\xi^\mu
+\pd{\Gamma^{\mu}_{\nu\sigma}}{x^\tau} \,  \xi^\tau \dot X^\nu \dot X^\sigma
+2\, \Gamma^{\mu }_{\nu\sigma}\, \dot X^\sigma \dot \xi^\nu.
\end{eqnarray}
That is, the right hand side is the differential operators that appears in the Jacobi equation. To
simplify the notation we will also write
$J^\mu_U[(\xi^\lambda)_{\lambda=0}^{n-1}]=J^\mu_U[\xi^\lambda]$.  We say that a curve $\xi\colon
I\to TU$ is a \emph{Jacobi field} in $U$ if $\pi\circ \xi\colon I\to U$ is a geodesic and
$J^\mu_U[\xi^\lambda]=0$
%$J^\mu_U[\xi]=0$
for all
$\mu\in \{0,\ldots, n-1\}$.

If  $(\widetilde U, \widetilde x^\mu)$ are overlapping
coordinates then
%and $\xi\colon I\to T(U\cap \widetilde U)$ we then have the
%transformation rule
\begin{eqnarray}
\label{eq:jacobiTrans}
J^\mu_{\widetilde U}\left[
\pd{\widetilde x^\lambda}{x^\rho} \xi^\rho
%\left.\pd{}{\widetilde x^\lambda}\right\vert_{X(s)}
\right] &=& \pd{\widetilde x^\mu}{ x^\sigma}   J^\sigma_U
\left[
\xi^\lambda %\right)%_{\lambda=0}^{n-1}
\right].
\end{eqnarray}
Thus, if we assume that $\xi^\mu$ are components for a vector field along $X$, then
the definition of a Jacobi field does not depend on the choice of coordinates.

There are also other ways to derive the Jacobi equation. One approach is to start with a
\emph{geodesic variation} around a geodesic $X\colon I\to M$. That is, a map
$\Lambda\colon (-\varepsilon_0,\varepsilon_0) \times\, I \to M$ such that
\begin{enumerate}
\item $s\mapsto \Lambda(\varepsilon,s)$ is a geodesic for each $\varepsilon\in (-\varepsilon_0, \varepsilon_0)$,
\item $\Lambda(0,s)=X(s)$.
\end{enumerate}
Then one can show that the tangent of $\Lambda$ in the $\varepsilon$-direction defines a vector
field $\xi\colon I\to TM$, $\xi(s) = \partial_\varepsilon \Lambda(\varepsilon,
s)\vert_{\varepsilon=0}$ along $X$, and moreover, the components of $\xi$ solves the Jacobi equation
\eqref{Jacobiequation}.  Conversely, any solutions to the Jacobi equation on a compact interval, can
be written as $\xi(s) = \partial_\varepsilon \Lambda(\varepsilon, s)\vert_{\varepsilon=0}$ for a
geodesic variation $\Lambda$. See for example \cite{BucataruDahl:2008}.
An advantage of this derivation is that $\xi$ a tangent vector by definition, and the derivation
does not involve any assumptions like $\xi^\mu$ and $\dot \xi^\mu$ beeing small or
infinitesimal.

%\begin{definition}
%\label{def:variation}
% Suppose $X\colon I\to M$ is a geodesic for an affine connection. Then a \emph{geodesic variation} of $X$ is a map
%\begin{align*}
%\Lambda:(-\varepsilon_0,\varepsilon_0)& \times\, I \to M
%\end{align*}
%such that $\varepsilon_0>0$ and
%\begin{enumerate}
%\item $s\mapsto \Lambda(\varepsilon,s)$ is a geodesic for each $\varepsilon\in (-\varepsilon_0, \varepsilon_0)$,
%\item $\Lambda(0,s)=X(s)$.
%\end{enumerate}
%\label{alowedvariation}
%\end{definition}
\subsection{The generalized Jacobi equation}
\label{sec:GenJacobiEq}
In the previous section we started with the exact geodesic deviation
equation \eqref{standardexactdeviationequation}, did a Taylor expansion of
$\Gamma^\mu_{\nu\sigma}(X+\xi)$ (equation \eqref{eq:gammaLin}) and assumed that $\xi^\mu$
and $\dot \xi^\mu$ are infinitesimal, so that
higher order terms
$\xi^\mu \xi^\nu$, $\dot \xi^\mu \xi^\nu$, $\dot \xi^\mu \dot \xi^\nu$, $\xi^\mu \xi^\nu
\xi^\sigma$, $\dot \xi^\mu \xi^\nu \xi^\sigma$, $\dot \xi^\mu \dot \xi^\nu \xi^\sigma$, $\ldots$ can
be neglected. This gives rise to the Jacobi equation \eqref{Jacobiequation}.
The generalized Jacobi equation %\eqref{eq:genJacobiEq}
is derived in the same way, but assuming
that only  $\xi^\mu$ is infinitesimal.
%$\xi^\mu
%\xi^\nu$, $\xi^\mu \xi^\nu \dot \xi^\sigma$, $\xi^\mu \xi^\nu \dot \xi^\sigma \xi^\rho$, $\ldots$
%can be neglected.
Under this approximation,
%
%If one linearizes respect to $\xi$, that is, if we assume that
%$\xi^\mu \xi^\nu$ is so small that one may assume it is zero, then
%
equation \eqref{standardexactdeviationequation} simplifies into the generalized Jacobi equation
in equation \eqref{eq:genJacobiEq}.

Suppose $(U, x^{\mu})$ are local coordinates for $M$, and $\xi^{0}, \ldots, \xi^{n-1}\colon I\to
\setR$ are functions along a geodesic $X\colon I\to U$.  As for the Jacobi equation we define
differential operators
\begin{eqnarray*}
%\label{eq:genJacobi}
 G_U^\mu[\xi^\lambda]\! &=& \!\ddot \xi^\mu
\,+\, \Gamma^\mu_{\nu\sigma} \left( 2 \dot \xi^\nu \dot X^\sigma + \dot \xi^\nu \dot \xi^\sigma \right)
\,+\, \pd{\Gamma^\mu_{\nu\sigma}}{x^\alpha} \xi^\alpha \left( \dot X^\nu + \dot \xi^\nu\right) \left( \dot X^\sigma + \dot \xi^\sigma\right),
\end{eqnarray*}
for $\mu\in \{0,\ldots, n-1\}$.
%
%\begin{definition}
%\label{def:genJacField}
%Suppose $(U,x^\mu)$ is a chart on $M$ and $X\colon I\to U$ a geodesic. Then
We will say that functions
$\xi^{0}\ldots, \xi^{n-1}\colon I\to \setR$ \emph{define a generalized Jacobi field in chart
  $U$} if $G^\mu_U[\xi^\lambda]=0$ for $\mu\in \{0,\ldots, n-1\}$.
%\end{definition}

Suppose $X$ is a geodesic $X\colon I\to M$. Then the Jacobi equation along $X$ is locally a linear
second order differential equation for components $\xi^{\mu}$.  Thus, for initial values
$\xi^\mu(s_0)$ and $\dot \xi^\mu (s_0)$ there exists a unique Jacobi field $\xi\colon I\to TM$ along $X$ with
these initial values.  See for example \cite{Lee:Curvature}. In contrast, the generalized Jacobi
equation $G^\mu_U[\xi^\lambda]=0$ is locally a \emph{non-linear} second order differential equation for
functions $\xi^{\mu}$.  The Picard-Lindel\"of theorem implies that for any $s_0\in I$, there exists
an neighborhood around $s_0$ such that the generalized Jacobi is uniquely solvable on this
neighborhood from initial values $\xi^{\mu}(s_0)$ and $\dot\xi^{\mu}(s_0)$. One can bound the size
of this neighborhood \cite[Theorem 1.1, Chapter 2]{Hartman:1964}. However, in general there is no
guarantee that a unique solution exists on the entire interval $I$ as for Jacobi fields.

For future reference, let us note that differential operators $G^{\mu}_U[\xi^\lambda]$ and
$J^{\mu}_U[\xi^\lambda]$ are related by
\begin{eqnarray}
\label{eq:GUdecomp}
    G^\mu_U[\xi^\lambda] &=& J^\mu_U[\xi^\lambda] + \Delta^\mu_U[\xi^\lambda],
\end{eqnarray}
where
\begin{eqnarray}
\label{eq:DeltaDef}
\Delta^\mu_U[\xi^\lambda] &=&  \Gamma^{\mu }_{\nu\sigma}\, \dot \xi^\nu \dot \xi^\sigma
+ 2 \pd{\Gamma^{\mu}_{\nu\sigma}}{x^\alpha} \,  \xi^\alpha \dot X^\nu \dot \xi^\sigma
+  \pd{\Gamma^{\mu}_{\nu\sigma}}{x^\alpha} \,  \xi^\alpha \dot \xi^\nu \dot \xi^\sigma.
\end{eqnarray}

\section{Coordinate invariance of the generalized Jacobi equation}
\label{sec:GJE-transformation}
As described in the introduction, this section contains the main results of this note.  As in
Section \ref{sec:EqsGeodesicDeviation} we assume that $M$ is a manifold with an affine and
torsion-free connection $\nabla$.

\begin{proposition}
\label{thm:pos}
Suppose $X\colon I\to M$ is a geodesic contained in a chart $(U, x^{\mu})$ and
functions $\xi^0, \ldots, \xi^{n-1}\colon I\to \setR$ define a generalized Jacobi field in
coordinates $x^\mu$.  If $\widetilde x^\mu$ are coordinates defined by the affine coordinate transformation
\eqref{eq:affineTrans}, then functions $\pd{\widetilde x^\mu}{x^\lambda} \xi^\lambda
\colon I\to
\setR$ define a generalized Jacobi field in coordinates $\widetilde x^\mu$.
\end{proposition}

\begin{proof}
Let $\ttt \xi^\mu = \pd{\widetilde x^\mu}{x^\sigma} \xi^\sigma$.
Equations
  \eqref{eq:transC_2}--\eqref{eq:transC_4} in Appendix \ref{app:transRules} show that
  $\pd{\widetilde x^\mu}{ x^\nu} \Delta^\nu_U[\xi^\lambda] = \Delta^\mu_{\widetilde U}[\ttt \xi^\lambda]$.
By equations \eqref{eq:jacobiTrans} and \eqref{eq:GUdecomp} we then have
 $\pd{\widetilde x^\mu}{ x^\nu} G^\nu_U[\xi^\lambda] = G^\mu_{\widetilde
    U}[\ttt \xi^\lambda]$ and the claim follows.
\end{proof}

\begin{proposition}
\label{thm:neg}
Suppose $M$ is a manifold of dimension $\ge 3$, $X\colon I\to M$ is a geodesic and $(U,x^\mu)$ are
local coordinates around $X(s_0)$ for some $s_0\in I$.  By shrinking $I$ to a neighborhood of $s_0$
we can find functions $\xi^\mu\colon I\to \setR$ and overlapping coordinates $(\widetilde U,
\widetilde x^\mu)$ around $X(s_0)$ such that $\xi^\mu$ define a generalized Jacobi field in chart
$U$, but $\pd{\widetilde x^\mu}{x^\nu} \xi^\nu$ does not define a generalized Jacobi field in chart
$\widetilde U$.
\end{proposition}

\begin{proof}
Let $p=X(s_0)$.
By \thCite{proposition} \ref{thm:pos} we may assume that $0\in \setR^n$
corresponds to $p$ in coordinates $x^\mu$.
Let $(\ttt U, \ttt x^\mu)$ be the coordinates determined by
\begin{eqnarray*}
\ttt x^{\mu}(x^0, \ldots,  x^{n-1}) &=& x^{\mu}+ \frac 1 6 T^{\mu}_{\tau\rho\sigma}  x^{\tau}  x^{\rho}  x^{\sigma},
\end{eqnarray*}
where
\begin{eqnarray*}
T^\mu_{\tau\rho \sigma} &=&
    \delta^\mu_\tau \delta_{\rho \sigma}
+ \delta^\mu_\rho \delta_{\tau \sigma}
+ \delta^\mu_\sigma\delta_{\tau\rho},
\end{eqnarray*}
and $\delta^\mu_\nu$ and $\delta_{\mu\nu}$ are the Kronecker delta symbols.

Let $u=\dot{X}(s_0)$.  Since $\dim M\ge 3$ we can find vectors $v,w \in T_{p}M$ such that
$\{u,v,w\}$ is an orthonormal basis with respect to the Euclidean metric $g=\delta_{\mu\nu}
dx^\mu\otimes dx^\nu$ on $U$.  By the Picard-Lindel\"of theorem we can shrink $I$ to a neighborhood
of $s_0$ and find functions $\xi^\mu \colon I \to \setR$ such that $\xi^\mu$ define a generalized
Jacobi field in chart $U$ and
$$
  \dot X^\mu(s_0)=u^\mu, \quad
  \xi\,^{\mu}(s_0) =  v^{\mu},
  \quad \dot{\xi}\,^{\mu} (s_0)=w^{\mu},
$$
when
%and such that if $\xi=\xi ^{\mu}(s)\left.\pd{}{x^{\mu}}\right\vert_{c(s)}$ for $s\in (-\delta, \delta)$,
$u= u^{\mu}\pd{}{ x^{\mu}}\vert_p$,
$v= v^{\mu}\pd{}{ x^{\mu}}\vert_p$
and $w= w^{\mu}\pd{}{ x^{\mu}}\vert_p$.
%%
%In coordinates $\ttt x^{\mu}$, let us also write
%$c(s)=(\widetilde x^i(s))$,
%$\xi=\ttt \xi
%^i(s)\left.\pd{}{\ttt x^i}\right\vert_{c(s)}$ for $t\in (-\varepsilon,
%\varepsilon)$,
%$u= \ttt u^{\mu}\pd{}{\ttt x^{\mu}}\vert_p$,
%$v= \ttt v^{\mu}\pd{}{\ttt x^{\mu}}\vert_p$
%and $w= \ttt w^{\mu}\pd{}{
% \ttt x^{\mu}}\vert_p$.
%
Let $\ttt \xi^\mu = \pd{\widetilde x^\mu}{x^\sigma} \xi^\sigma$.
% and let also $X^\mu$ and $\ttt X^\mu$ be
%components for $\dot X(s_0)$ in coordinates $x^\mu$ and $\widetilde x^\mu$, respectively.

Contracting $G^{\nu}_{\ttt U}[\ttt \xi^\lambda ]=J^{\nu}_{ \ttt U}[\ttt \xi^\lambda ] +\Delta^{\nu}_{\ttt U}[\ttt \xi^\lambda]$  %equation \eqref{eq:GUdecomp}
by $\pd{ x^{\mu}}{\ttt x^{\nu}}$, applying equations \eqref{eq:jacobiTrans}--\eqref{eq:GUdecomp}
%we have
% with notation as in equations \eqref{eq:Jdef} and
%\eqref{eq:DeltaDef}. xxxx
%
and equations \eqref{eq:transC_2}--\eqref{eq:transC_4} in Appendix
\ref{app:transRules} yields
\begin{eqnarray*}
  \pd{ x^{\mu}}{\ttt x^{\nu}}
G^{\nu}_{ \ttt U}[\ttt \xi^\lambda]
&=&  J^{\mu}_{  U}[\xi^\lambda] +
 \pd{  x^{\mu}}{ \ttt x^{\nu}} \Delta^{\nu}_{ \ttt U}[\ttt \xi^\lambda] \\
&=&\pd{ x^{\mu}}{\ttt x^{\nu}} \Delta^{\nu}_{ \ttt U}[\ttt \xi^\lambda] - \Delta^{\mu}_{ U }[\xi^\lambda] \\
&=& - \pd{x^\mu}{\ttt x^\nu} \frac{\partial ^3 \ttt x^{\nu}}{\partial  x^{\tau}\partial x^{\rho}\partial x^{\sigma}}  \xi^{\tau} \left( 2\dot{X}\, ^{\rho} + \dot {\xi}\,^{\rho}\right)  \dot {\xi}\,^{\sigma},
\end{eqnarray*}
where all expressions are evaluated at $s_0$.
Since $T^{\mu}_{\tau\rho\sigma}$ is symmetric in $\tau\rho\sigma$,
we have
$ \frac{\partial ^3 \ttt x^{\mu}}{\partial  x^{\tau}\partial x^{\rho}\partial x^{\sigma}}  = T^{\mu}_{\tau\rho\sigma}$.
Moreover, since $\pd{x^{\mu}}{\ttt x^{\nu}}(p) = \delta^{\mu}_{\nu}$ and since $u,v,w$ are orthonormal it follows that
\begin{eqnarray*}
G^{\mu}_{ \ttt U}[\ttt{\xi}\,\!^\lambda] &=& -g(2u+w, w)\,  v^{\mu}
-g(v, w)\, (2u^{\mu}  +  w^{\mu} )
-g(v, 2u+ w) \,  w^{\mu}  \\
&=& -v^{\mu}
\end{eqnarray*}
at $s_0$. We have shown that $G^\mu_{\ttt U}[\ttt \xi^\lambda]\neq 0$ at $s_0$. Thus functions $\ttt \xi^\mu$
do not define a generalized Jacobi field in chart $\widetilde U$, and the claim follows.
\end{proof}

\subsection*{Acknowledgements}
MD was funded by the Academy of Finland (project 13132527) and by the Institute of Mathematics at
Aalto University. RGT was funded by FAPESP, process 2010/11934-6.

\appendix

\section{Coordinate transformations}
\label{app:transRules}
Suppose $X$ is a curve $X\colon I\to M$,  and
$x^{\mu}$ and $\ttt x^{\mu}$ are coordinates around $X(s_0)$ with $\pdd{
\ttt x^{\mu}}{ x^{\nu}}{ x^{\rho}}(X(s_0))=0$ for some $s_0\in I$.
If $\xi^\mu$ are functions $\xi^\mu\colon I\to \setR$ and functions $\ttt \xi^\mu$ are defined by
$\ttt \xi^{\mu} =  \pd{\ttt x^{\mu}}{ x^{\nu}}{ \xi}^{\nu}$, then at $X(s_0)$ we have transformation rules
\begin{eqnarray}
\dot \xi^{\mu} &=& \label{eq:transC_2}
\pd{x^{\mu}}{\widetilde x^{\nu}}   \dot {\widetilde  \xi}\,\!^{\nu},\\
\pd{\widetilde x^{\mu}}{x^{\sigma}}    \Gamma^{\sigma}_{\rho\nu}       &=& \label{eq:transC_3}
\widetilde \Gamma^{\mu}_{\sigma\lambda}   \pd{\widetilde x^{\sigma}}{x^{\rho}}  \pd{\widetilde x^{\lambda}}{x^{\nu}}, \\
\pd{\widetilde x^{\mu}}{x^{\sigma}}  \pd{  \Gamma^{\sigma}_{\nu\rho}}{x^{\lambda}}       &=& \label{eq:transC_4}
\pd{\widetilde \Gamma^{\mu}_{\epsilon\sigma} }{\widetilde x^{\delta}} \pd{\widetilde x^{\delta}}{x^{\lambda}} \pd{\widetilde x^{\sigma}}{x^{\nu}}  \pd{\widetilde x^{\epsilon}}{x^{\rho}}
+  \frac{\partial ^3 \ttt x^{\mu}}{\partial  x^{\nu}\partial x^{\rho}\partial x^{\lambda}} .
\end{eqnarray}
%Equations \eqref{eq:transC_1}--\eqref{eq:transC_4}
These equations show that between coordinates $x^{\mu}$ and $\widetilde x^{\mu}$, objects $\xi^{\mu},
\dot \xi^{\mu}$, $\Gamma^{\mu}_{\rho \nu}$ and $\pd{\Gamma^{\mu}_{\rho\nu}}{x^{\sigma}}$ transform
as tensors at $X(s_0)$ (up to an extra term in the last equation).  For a general coordinate
transformation, these transformation rules %equations \eqref{eq:transC_3}--\eqref{eq:transC_4}
are considerably more involved.
See equation \eqref{eq:gammaSymbTrans}.

%\bibliographystyle{amsalpha}
%\bibliography{references}

\providecommand{\bysame}{\leavevmode\hbox to3em{\hrulefill}\thinspace}
\providecommand{\MR}{\relax\ifhmode\unskip\space\fi MR }
\providecommand{\MRhref}[2]{%
  \href{http://www.ams.org/mathscinet-getitem?mr=#1}{#2}
}
\providecommand{\href}[2]{#2}

\end{document}